\pdfoutput=1 
\documentclass[letterpaper, 10 pt, conference, onecolumn]{ieeeconf} 

\IEEEoverridecommandlockouts                              
\usepackage{amsmath}
\usepackage{xcolor}
\usepackage{amssymb}
\usepackage{algorithm}
\usepackage{graphicx}
\usepackage{algpseudocode}
\usepackage{cite} % Recommended for handling multiple citations

\usepackage{amsthm}   % Now safe to load amsthm

\theoremstyle{definition} % Now this command will work
\newtheorem{definition}{Definition}
\newtheorem{remark}{Remark}
\newtheorem{lemma}{Lemma}
\newtheorem{assumption}{Assumption}
\newtheorem{problem}{Problem}

\theoremstyle{plain}      % Switch back to italics for theorems
\newtheorem{theorem}{Theorem}

\title{\LARGE \bf
Willems' Fundamental Lemma with Large Noisy Fragmented Dataset
}

\author{Sahand Kiani and Constantino M. Lagoa% <-this % stops a space
\thanks{The authors are with the Department of Electrical and Computer Engineering, Pennsylvania State University, State College, PA 16801, USA. S. Kiani is a PhD student (e-mail: szk6437@psu.edu). C. M. Lagoa is a full professor (e-mail: cml18@psu.edu).}%
}

\begin{document}

\maketitle
\thispagestyle{empty}
\pagestyle{empty}

%%%%%%%%%%%%%%%%%%%%%%%%%%%%%%%%%%%%%%%%%%%%%%%%%%%%%%%%%%%%%%%%%%%%%%%%%%%%%%%%
\begin{abstract}
Willems' Fundamental Lemma enables parameterizing all trajectories generated by a Linear Time-Invariant (LTI) system directly from data. However, this lemma relies on the assumption of noiseless measurements. In this paper, we provide an approach that enables the applicability of Willems' Fundamental Lemma with a large noisy-input, noisy-output fragmented dataset, without requiring prior knowledge of the noise distribution. We introduce a computationally tractable and lightweight algorithm that, despite processing a large dataset, executes in the order of seconds to estimate the invariants of the underlying system, which is obscured by noise. The simulation results demonstrate the effectiveness of the proposed method.
\end{abstract}


%%%%%%%%%%%%%%%%%%%%%%%%%%%%%%%%%%%%%%%%%%%%%%%%%%%%%%%%%%%%%%%%%%%%%%%%%%%%%%%%
\section{Introduction}
Data-driven control has emerged as a powerful paradigm for analyzing system properties and designing controllers directly from collected data, bypassing the need for explicit analytical models \cite{de2019formulas}\cite{zheng2025robust}. Central to many of these methodologies is Willems' Fundamental Lemma \cite{willems2005note}\cite{9062331}, which establishes that a persistently exciting trajectory can span all possible finite-horizon trajectories generated by the LTI system. This lemma assumes that the data must be free of noise. To overcome this limitation, scholars have proposed various regularization techniques, as reviewed below.

In behavioral systems theory, measurement noise is traditionally mitigated through structured low-rank approximations or matrix regularization \cite{markovsky2021behavioral}\cite{willems1997introduction}. Alternatively, robust data-driven predictive control frameworks often avoid statistical assumptions entirely by treating bounded additive noise as a system uncertainty. This is achieved by mapping the noise to the multiplicative uncertainty of the model \cite{berberich2020data}, or formulating min-max optimizations to ensure the resilience of the controller against worst-case disturbances \cite{huang2023robust}. Although these robust formulations successfully stabilize the closed-loop system, they are conservative and rely on worst-case relaxations. Finally, to specifically address simultaneous input-output measurement noise—often referred to as the errors-in-variables problem—recent methodologies construct auxiliary state-space representations from time-shifted data. For example, the authors in \cite{li2026controller} design dynamic output-feedback controllers by formulating robust stabilization problems over these auxiliary systems, guaranteeing stability for all possible dynamics consistent with an assumed noise bound.

In contrast to previous approaches, this work leverages the properties of large datasets to recover the system invariants directly from noisy measurements. By utilizing an ensemble of multiple experiments, we secure the data realizations necessary to decouple the stochastic noise and extract the underlying system invariants.

\subsection{Main Contribution}
The primary contribution of this paper is an algorithm that enables the applicability of Willems' Fundamental Lemma when a large noisy fragmented dataset is available, without requiring prior knowledge of the noise distributions. Crucially, because this framework processes data across an ensemble rather than relying on a single, long continuous trajectory, it easily accommodates fragmented data. This includes handling missing data segments occurring in the middle of a recording session, as well as arbitrary temporal gaps between completely distinct trials, regardless of their varying initial conditions. Our approach begins with an analysis of the left null space of stacked input-output Hankel matrices within a multi-experiment setting. Building upon this, we demonstrate that by leveraging the Law of Large Numbers (LLN) across a sufficient number of experiments, the underlying system invariants can be estimated purely from noisy data. Finally, by formulating an optimization problem that can be solved using Singular Value Decomposition (SVD), the proposed algorithm estimates the invariants of the system. Furthermore, although our framework relies on substantially larger datasets compared to traditional approaches in systems theory, we demonstrate that the computational complexity of the proposed invariant recovery algorithm remains remarkably lightweight.

\textbf{Notation:} For a matrix $X \in \mathbb{R}^{p \times q}$, $\text{Im}(X)$ denotes its image (or column space), defined as $\text{Im}(X) := \{ y \in \mathbb{R}^p \mid \exists x \in \mathbb{R}^q, y = Xx \}$. The left kernel (or left null space) of $X$ is denoted by $\mathcal{N}(X)$ and is defined as $\mathcal{N}(X) := \{ v \in \mathbb{R}^p \mid v^\top X = \mathbf{0}^\top \}$. The element in the $r$-th row and $s$-th column of a matrix $X$ is denoted by $X_{(r,s)}$, and the subscript notation $X_{[r, j:k]}$ is used to denote the row vector extracted from the $r$-th row of $X$ spanning columns $j$ through $k$. The vector of all ones of size $n$ is denoted by $\mathbf{1}_n \in \mathbb{R}^n$. Given a sequence of data $\{\zeta_k\}_{k=0}^{N-1}$ and a chosen depth $L$, we define the Hankel matrix $H_{L}(\zeta)$ as:

\begin{equation}
H_{L}(\zeta) := \begin{bmatrix}
\zeta_{0} & \zeta_{1} & \cdots & \zeta_{N-L} \\
\zeta_{1} & \zeta_{2} & \cdots & \zeta_{N-L+1} \\
\vdots & \vdots & \ddots & \vdots \\
\zeta_{L-1} & \zeta_{L} & \cdots & \zeta_{N-1}
\end{bmatrix}.
\end{equation}

\section{Willems' Fundamental Lemma in the Multi-Experiment Setting}\label{section2}

Consider the following discrete-time LTI system:
\begin{equation}
\begin{aligned}
x_{k+1} &= Ax_{k}+Bu_{k}, \quad x_{0} \in \mathbb{R}^n, \\
y_{k} &= Cx_{k}+Du_{k},
\end{aligned}\label{sys}
\end{equation}
where $x_{k} \in \mathbb{R}^{n}$ is the state, $u_{k} \in \mathbb{R}^{m}$ is the input and $y_{k} \in \mathbb{R}^{p}$ is the output. Throughout the paper, the matrices $(A,B,C,D)$ are unknown, while the dimensions of the system $n,m,p$ are known.

\begin{definition} \cite{willems2005note} 
The input signal $\{u_{k}\}_{k=0}^{N-1}$ with $u_{k}\in\mathbb{R}^{m}$ is persistently exciting of order $L$ if $\text{rank}(H_{L}(u))=mL$.
\end{definition}

We consider the availability of a data set comprising $N_t$ experiments, where each experiment $i$ yields an input-output trajectory of length $N$, such that:
\begin{equation}
\mathcal{D} := \left\{ \{u_{k}^{(i)}, y_{k}^{(i)}\}_{k=0}^{N-1} \right\}_{i=0}^{ N_t-1}.
\end{equation}

\begin{remark}
The multi-experiment formulation naturally encapsulates the concept of fragmented data. If a temporal gap or sensor failure occurs in the middle of a continuous data collection process, the recorded data is simply divided into separate experiments. Consequently, the duration of the gaps between these fragments is entirely irrelevant to the proposed framework.
\end{remark}

\begin{assumption}\label{AssumpData}
For each experiment $i$, the input sequence $\{u_k^{(i)}\}_{k=0}^{N-1}$ is persistently exciting of order $L+n$.
\end{assumption}

Note that while we adopt this standard assumption for theoretical clarity, the multi-experiment setting of our proposed method inherently allows this requirement to be relaxed to a collective excitation property across all $N_t$ fragments.

\begin{definition}\cite{berberich2020trajectory}
A sequence $\{u_{k}^{(i)},y_{k}^{(i)}\}_{k=0}^{N-1}$ is a trajectory of an LTI system $G$ for a given experiment $i$ if there exists an initial state $x_0^{(i)}\in\mathbb{R}^{n}$ and a state sequence $\{x_{k}^{(i)}\}_{k=0}^{N}$ satisfying the state-space equations for $k=0,...,N-1$, where $(A,B,C,D)$ is a minimal realization of $G$.
\end{definition}

\begin{theorem}\cite{berberich2020trajectory}\label{thm:willems_lemma}
  Suppose $\{u_{k}^{(i)},y_{k}^{(i)}\}_{k=0}^{N-1}$ is a trajectory of an LTI system $G$ for a given experiment $i$ under Assumption \ref{AssumpData}. Then, a sequence $\{\overline{u}_{k},\overline{y}_{k}\}_{k=0}^{L-1}$ is a trajectory of $G$ if and only if there exists a vector $\alpha^{(i)}\in\mathbb{R}^{N-L+1}$ such that
\begin{equation} \label{WL}
\mathbb{H}^{(i)} \alpha^{(i)} = z, \quad \mathbb{H}^{(i)} := \begin{bmatrix} H_{L}(u^{(i)}) \\ H_{L}(y^{(i)}) \end{bmatrix}, \ z := \begin{bmatrix} \overline{u} \\ \overline{y} \end{bmatrix}.
\end{equation}
\end{theorem}

Theorem \ref{thm:willems_lemma} establishes that for each experiment $i$, the range space $\text{Im}(\mathbb{H}^{(i)})$ spans any $z$ generated by (\ref{sys}). Because the dynamics of the system imposes linear dependencies between the inputs and outputs, each $\mathbb{H}^{(i)}$ is row-rank deficient. 

To characterize the system invariants, rather than working with the range space, we focus on the orthogonal complement: the left null space $\mathcal{N}(\mathbb{H}^{(i)})$. By Theorem \ref{thm:willems_lemma}, a candidate vector $z \in \mathbb{R}^{(m+p)L}$ is a trajectory of (\ref{sys}) if and only if $\forall i \in \{0, 1, \dots, N_t-1\}$, it is strictly orthogonal to this left null space:
\begin{equation}
\mathbb{H}^{(i)} \alpha^{(i)} = z \iff v^{\top} z = 0, \quad \forall v \in \mathcal{N}(\mathbb{H}^{(i)}).
\end{equation}

This equivalence is a direct consequence of fundamental linear algebra principles—specifically, that the column space and left null space of a matrix are orthogonal complements—and its formal proof is therefore omitted. Since all inputs in each experiment are assumed to be persistently exciting, this left null space $\mathcal{N}(\mathbb{H}^{(i)})$ has an exact dimension of $pL - n$.

\begin{remark}
Under Assumption \ref{AssumpData}, each experiment fully captures the invariants of the underlying LTI system (\ref{sys}), which we characterize here by the left null space $\mathcal{N}(\mathbb{H}^{(i)})$. Consequently, all $N_t$ experiments share the same left null space.
\end{remark}

Let us define $\mathbf{M}_{\mathcal{D}}$, the aggregate correlation matrix across all experiments as follows:
\begin{equation} \label{eq:correlation_matrix}
\mathbf{M}_{\mathcal{D}} := \frac{1}{N_t}\sum_{i=0}^{N_t-1} \mathbb{H}^{(i)} (\mathbb{H}^{(i)})^\top.
\end{equation}

The symmetric positive semi-definite matrix $\mathbb{H}^{(i)} (\mathbb{H}^{(i)})^\top$ is chosen to directly target the left null space of $\mathbb{H}^{(i)}$.

\begin{lemma} \label{lem:noiseless_kernel}
Consider the LTI system (\ref{sys}). Under Assumption \ref{AssumpData}, a vector $v \in \mathbb{R}^{(m+p)L}$ satisfies the following condition:
\begin{equation}
    v^{\top} \mathbf{M}_{\mathcal{D}} v = 0, \label{vec} 
\end{equation}
if and only if $v$ lies in the intersection of all left null spaces in the experiments.
\end{lemma}

\begin{proof}
Please find the proof in the Appendix.
\end{proof}

As established in Lemma \ref{lem:noiseless_kernel}, finding system invariants using $\mathcal{D}$ is equivalent to finding a vector $v$ that lies in the null space of the symmetric positive semi-definite aggregate matrix $\mathbf{M}_{\mathcal{D}}$. Although computing this invariant vector is trivial with noiseless data, the presence of measurement noise structurally obscured the system invariants.

\section{Problem Formulation}

In real-world scenarios, measurements are inevitably corrupted by noise. We consider that the true inputs and outputs are affected by additive noise, yielding the available measurements $\tilde{u}_{k}, \tilde{y}_{k}$, defined as follows:
\begin{equation}\label{noisemeas}
\tilde{u}_{k} = u_{k} + \eta^u_{k}, \quad \tilde{y}_{k} = y_{k} + \eta^y_{k},
\end{equation}
where $\eta^u_{k} \in \mathbb{R}^{m}$ and $\eta^y_{k} \in \mathbb{R}^{p}$. Therefore, the available large fragmented noisy data set is defined as follows: 
\begin{equation} \label{eq:noisy_M}
\tilde{\mathcal{D}} := \{ \{\tilde{u}_{k}^{(i)}, \tilde{y}_{k}^{(i)}\}_{k=0}^{N-1} \}_{i=0}^{ N_t-1}.
\end{equation}

\begin{assumption}\label{noise_bound}
The mutually independent noise sequences are i.i.d. over time and drawn from unknown distributions. Furthermore, it is assumed that all statistical moments up to the fourth order exist and are finite. The true input-output data are bounded; that is, there exist $J_{u}, J_{y}$ such that $\|u\| \leq J_{u}, \|y\| \leq J_{y}$.
\end{assumption}

As will be demonstrated, Assumption \ref{noise_bound} is necessary to theoretically guarantee the convergence of the proposed algorithm. Additive noise obscures the left null space of the underlying system and motivates the following problem.

\begin{problem}\label{prob:1}
Under Assumptions \ref{AssumpData} and \ref{noise_bound}, given the noisy dataset $\tilde{\mathcal{D}}$ and without prior knowledge of the noise distributions, recover $\mathcal{N}(\mathbb{H}^{(i)})$.
\end{problem}

\section{Estimation of System Invariants from Large Noisy Fragmented Dataset}
To address Problem \ref{prob:1}, we seek to establish a formal mathematical relationship connecting $\mathbf{M}_\mathcal{D}$ with the available noisy data.

\begin{remark}
For simplicity of presentation, we assume that the input and output noise distributions share identical first and second moments, denoted as $m_1$ and $m_2$, respectively. However, this framework naturally extends to the general case where the input and output measurements are governed by distinct noise distributions with separate moments ($m_1^u, m_1^y$ and $m_2^u, m_2^y$).\label{rem:identical_moments}
\end{remark}

To this end, we begin by examining $\mathbb{E}[\mathbf{M}_\mathcal{D}] = \mathbf{M}_\mathcal{D}$. Since the underlying system and the applied input sequences are entirely deterministic, this relationship strictly simplifies to:
\begin{equation}\label{MD}
\mathbb{E}[\mathbf{M}_\mathcal{D}] = \frac{1}{N_t}\sum_{i=0}^{N_t-1} \mathbb{E}\big[\mathbb{H}^{(i)}(\mathbb{H}^{(i)})^\top\big] = \frac{1}{N_t}\sum_{i=0}^{N_t-1} \mathbb{H}^{(i)}(\mathbb{H}^{(i)})^\top.
\end{equation}

Let $N_c := N-L+1$ denote the number of columns in $\mathbb{H}^{(i)}$. We define the $r$-th row of the true and noisy stacked Hankel matrices for the $i$-th experiment, respectively, as:
\begin{equation}\label{law1}
h_r^{(i)} := \mathbb{H}^{(i)}_{[r, 1:N_c]}, \quad \tilde{h}_r^{(i)} = h_r^{(i)} + \eta_r^{(i)},
\end{equation}
where $\eta_r^{(i)} \in \mathbb{R}^{1 \times N_c}$ represents a noise row vector whose elements consist of i.i.d. measurement noise corrupting the true sequence $h_r^{(i)}$, with first and second moments $m_1$ and $m_2$ (as established in Remark \ref{rem:identical_moments}). The scalar elements of the matrix $\mathbb{H}^{(i)}(\mathbb{H}^{(i)})^\top$ are dictated by the inner products of these rows (i.e., multiplying the $r$-th row of $\mathbb{H}^{(i)}$ by the $s$-th column of $(\mathbb{H}^{(i)})^\top$). Therefore, the following holds:
\begin{equation}\label{rule1}
h_r^{(i)} = \mathbb{E}[h_r^{(i)}], \quad h_r^{(i)} (h_s^{(i)})^\top = \mathbb{E}[h_r^{(i)} (h_s^{(i)})^\top].
\end{equation}

By substituting the relationship $h_r^{(i)} = \tilde{h}_r^{(i)} - \eta_r^{(i)}$, we can explicitly compute the true row and expand the uncorrupted inner product $h_r^{(i)} (h_s^{(i)})^\top$ as follows:
\begin{equation} \label{eq:moments_expansion}
\begin{aligned}
h_r^{(i)} &= \mathbb{E}[\tilde{h}_r^{(i)}] - m_1 \mathbf{1}_{N_c}^\top, \\
h_r^{(i)} (h_s^{(i)})^\top &= \mathbb{E}\left[(\tilde{h}_r^{(i)} - \eta_r^{(i)})(\tilde{h}_s^{(i)} - \eta_s^{(i)})^\top\right].
\end{aligned}
\end{equation}

By expanding the terms in (\ref{eq:moments_expansion}) using the aforementioned expectation rules and substituting the results into in (\ref{MD}), we can explicitly express the elements of the true $\mathbf{M}_\mathcal{D}$ strictly as a function of the expected noisy measurements and the unknown noise moments:

\begin{equation} \label{eq:matrix_reconstruct}
\begin{split}
&(\mathbf{M}_{\mathcal{D}})_{(r,s)} = \\
&\quad \begin{cases} 
\begin{aligned}
&\frac{1}{N_t} \sum_{i=0}^{N_t-1} \mathbb{E}\big[\tilde{h}_r^{(i)} (\tilde{h}_r^{(i)})^\top\big] \\
&\quad - 2 m_1 \Big( \frac{1}{N_t} \sum_{i=0}^{N_t-1} \mathbb{E}\big[\tilde{h}_r^{(i)}\big] \Big) \mathbf{1}_{N_c} \\
&\quad + N_c(2m_1^2 - m_2),
\end{aligned} & \text{if } r = s \\[4ex]
\begin{aligned}
&\frac{1}{N_t} \sum_{i=0}^{N_t-1} \mathbb{E}\big[\tilde{h}_r^{(i)} (\tilde{h}_s^{(i)})^\top\big] \\
&\quad - m_1 \Big( \frac{1}{N_t} \sum_{i=0}^{N_t-1} \mathbb{E}\big[\tilde{h}_r^{(i)} + \tilde{h}_s^{(i)}\big] \Big) \mathbf{1}_{N_c} \\
&\quad + N_c m_1^2,
\end{aligned} & \text{if } r \neq s
\end{cases}
\end{split}
\end{equation}

As shown in (\ref{eq:matrix_reconstruct}), the element-wise reconstruction of $\mathbf{M}_{\mathcal{D}}$ depends on the expected values of the noisy measurements and moments of the noise, which are inherently inaccessible in practice. To bridge this gap, we first evaluate these expected values using empirical sample means computed across all $N_t$ experiments. For the corresponding row indices $r$ and $s$, these sample means are:
\begin{equation}\label{empmeans}
\frac{1}{N_t} \sum_{i=0}^{N_t-1} \tilde{h}_r^{(i)}, \quad \frac{1}{N_t} \sum_{i=0}^{N_t-1} \tilde{h}_r^{(i)} (\tilde{h}_s^{(i)})^\top.
\end{equation}

By directly substituting the expectations in (\ref{eq:matrix_reconstruct}) with their corresponding empirical sample means (\ref{empmeans}), we define the empirical estimator matrix, denoted as $\hat{\mathbf{M}}_\mathcal{D}$, as follows:

\begin{equation} \label{eq:M_hat_def}
\begin{split}
&(\hat{\mathbf{M}}_{\mathcal{D}})_{(r,s)} := \\
&\begin{cases} 
\begin{aligned}
&\frac{1}{N_t} \sum\nolimits_{i=0}^{N_t-1} \tilde{h}_r^{(i)} (\tilde{h}_r^{(i)})^\top \\[1ex]
&\quad - 2m_1 \Big( \frac{1}{N_t} \sum\nolimits_{i=0}^{N_t-1} \tilde{h}_r^{(i)} \Big)\mathbf{1}_{N_c} \\[1ex]
&\quad + N_c(2m_1^2 - m_2),
\end{aligned} & \text{if } r = s \\[4ex]
\begin{aligned}\\
&\frac{1}{N_t} \sum\nolimits_{i=0}^{N_t-1} \tilde{h}_r^{(i)} (\tilde{h}_s^{(i)})^\top \\[1ex]
&\quad - m_1 \Big( \frac{1}{N_t} \sum\nolimits_{i=0}^{N_t-1} \big( \tilde{h}_r^{(i)} + \tilde{h}_s^{(i)} \big) \Big)\mathbf{1}_{N_c} \\[1ex]
&\quad + N_c m_1^2,
\end{aligned} & \text{if } r \neq s
\end{cases}
\end{split}
\end{equation}

The following theorem formally establishes that as the number of experiments $N_t \to \infty$, our estimator matrix $\hat{\mathbf{M}}_\mathcal{D}$ rigorously converges to the underlying matrix $\mathbf{M}_\mathcal{D}$.

\begin{theorem} \label{thm:lln_convergence}
Under Assumptions \ref{AssumpData} and \ref{noise_bound}, by the Law of Large Numbers \cite{grimmett2020probability}, the empirical estimator $\hat{\mathbf{M}}_\mathcal{D}$ converges entry-wise almost surely to $\mathbf{M}_\mathcal{D}$ as the number of experiments $N_t \to \infty$, with a convergence rate of $\mathcal{O}(1/N_t)$. That is:
\begin{equation}
\lim_{N_t \to \infty} \left( \hat{\mathbf{M}}_\mathcal{D} - \mathbf{M}_{\mathcal{D}} \right) = \mathbf{0},\quad\text{a.s.}
\end{equation}
\end{theorem}

\begin{proof}
    Please find the proof in the Appendix.
\end{proof}

Following Theorem \ref{thm:lln_convergence}, as $N_t \to \infty$, the system invariants (Problem \ref{prob:1}) can be recovered numerically by evaluating the SVD of $\hat{\mathbf{M}}_\mathcal{D}$ at the noise moments $(m_1^*, m_2^*)$. By defining a sufficiently small threshold $\epsilon_\sigma > 0$, the singular vectors associated with the minimum singular values below $\epsilon_\sigma$ can potentially span the underlying left null space. However, since $(m_1^*, m_2^*)$ are strictly unknown in practice, we formulate a procedure to simultaneously recover both the noise moments and the underlying left null space.

To achieve this, we divide the estimation into two distinct algorithms. Algorithm \ref{alg:matrix_construction} establishes the construction of matrix $\hat{\mathbf{M}}_{\mathcal{D}}(m_1,m_2)$ as follows: 

\begin{algorithm}[ht]
\caption{Construction of $\hat{\mathbf{M}}_{\mathcal{D}}(m_1,m_2)$}
\label{alg:matrix_construction}
\begin{algorithmic}[1]
\Require $\tilde{\mathcal{D}}, L, m, n, p$
\State Compute the empirical averages for all pairs of row indices $(r, s)$ using (\ref{empmeans}) from $\tilde{\mathcal{D}}$
\State Construct $\hat{\mathbf{M}}_\mathcal{D}(m_1,m_2)$ using (\ref{eq:M_hat_def})
\State \Return Parameterized $\mathbf{M}_\mathcal{D}(m_1, m_2)$ in moments
\end{algorithmic}
\end{algorithm}

\begin{remark}
Although satisfying the convergence condition in Theorem \ref{thm:lln_convergence} requires a large data set, the data aggregation step in Algorithm \ref{alg:matrix_construction} is computationally efficient. The extraction of empirical means and cross-correlations scales linearly with the number of experiments, $\mathcal{O}(N_t)$, and consists of highly parallelizable vector operations. Crucially, this $\mathcal{O}(N_t)$ aggregation is performed only once. 
\end{remark}

Since the moments of the noise, $m^*_1$ and $m^*_2$, are completely unknown, we perform a grid search over a range of these parameters. For each candidate pair $(m_1, m_2)$ in the considered grid, we evaluate the SVD of $\hat{\mathbf{M}}_{\mathcal{D}}(m_1,m_2)$. 

\begin{remark}
In the general case where input and output measurements are governed by distinct noise distributions, this two-dimensional parameter search naturally expands into a four-dimensional grid search over $(m_1^u, m_2^u, m_1^y, m_2^y)$.
\end{remark}

Building upon this, Algorithm \ref{alg:recovery} systematically explores the search grid to locate the noise parameters that drive the minimum singular value closest to zero, which leads to estimating the noise moments and the left null space of system invariants.
\begin{algorithm}[ht]
\caption{Estimation of System Invariants}
\label{alg:recovery}
\begin{algorithmic}[1]
\State \textbf{Input:} $\hat{\mathbf{M}}_\mathcal{D}(m_1, m_2)$, grids $\mathcal{G}_{m_1}, \mathcal{G}_{m_2}$, threshold $\epsilon_\sigma$
\State Initialize $\mathcal{S} \leftarrow \emptyset$, and $\hat{V} \leftarrow \mathbf{0}\in\mathbb{R}^{(m+p)L \times (pL-n)}$
\For{$(m_1, m_2) \in \mathcal{G}_{m_1} \times \mathcal{G}_{m_2}$}
    \State Compute SVD of $\hat{\mathbf{M}}_\mathcal{D}(m_1, m_2)$ to obtain singular values $\Sigma$ and associated singular vectors $V$
    \If{exactly $pL - n$ minimum singular values in $\Sigma$ are $< \epsilon_\sigma$}
        \State Extract $pL - n$ singular vectors into $\hat{V}$
        \State $\mathcal{S} \leftarrow \mathcal{S} \cup \{(m_1, m_2, \hat{V})\}$
    \EndIf
\EndFor
\State Extract the tuple $(m_1^*, m_2^*, \hat{V})$ from $\mathcal{S}$ with the smallest minimum singular value
\State \textbf{Output:} $m_1^*, m_2^*$ and estimated left null space basis $\hat{V}$
\end{algorithmic}
\end{algorithm}

As described in Algorithm \ref{alg:recovery}, as $N_t \to \infty$, multiple candidate pairs $(m_1, m_2)$ can yield a minimum singular value equal to zero. In practice, the tolerance threshold $\epsilon_\sigma$ is chosen to be as small as possible to strictly separate singular values closest to zero. To isolate the best possible estimate of the system invariants, the algorithm first filters the set of candidates by enforcing the known theoretical nullity of the LTI system. Among these valid candidates, the algorithm explicitly selects the pair with the smallest minimum singular value.

\begin{remark}
The computational complexity of Algorithm \ref{alg:recovery} is strictly independent of the dataset size $N_t$. For search grids with $N_g$ evaluation points per dimension, the algorithm computes the minimum singular value of the fixed $(m+p)L \times (m+p)L$ matrix $\hat{\mathbf{M}}_\mathcal{D}$ at most $N_g^2$ times. Because the data aggregation over $N_t$ experiments is precomputed in Algorithm \ref{alg:matrix_construction}, the overall time complexity of Algorithm \ref{alg:recovery} is bounded by $\mathcal{O}(N_g^2 ((m+p)L)^3)$.
\end{remark}

\section{Simulation Results}

To validate the effectiveness of the proposed algorithms, we consider a discrete-time LTI system ($n=3, m=2, p=3$) using the following state-space matrices:
\begin{equation*}
A = \left[\begin{smallmatrix} 0.8 & -0.1 & 0.0 \\ 0.1 & 0.7 & 0.1 \\ 0.0 & -0.2 & 0.6 \end{smallmatrix}\right], \ 
B = \left[\begin{smallmatrix} 1.0 & 0.0 \\ 0.0 & 1.0 \\ 0.5 & 0.5 \end{smallmatrix}\right], \ 
C = I_3, \ 
D = \mathbf{0}_{3 \times 2}.
\end{equation*}
These matrices are assumed to be unknown and are only used to generate data.

\textit{Data Generation:} We conducted $N_t = 10,000$ experiments, exciting the system with persistently exciting input sequences of length $N = 30$ and constructing Hankel matrices of depth $L = 2$. The input-output measurements are corrupted by additive i.i.d. noise with a mean $m_1 = 1.0$ and standard deviation $\sigma = 2.0$ (yielding $m_2 = 5.0$).

\textit{Simulation Environment}: The simulations were performed on a Mac Mini powered by an Apple M4 chip with 32 GB of RAM. The algorithms were implemented in Python.

\textit{Experimental Results:} Once the data set is collected, we employ Algorithm \ref{alg:matrix_construction} to construct $\hat{\mathbf{M}}_\mathcal{D}(m_1, m_2)$, which executes in $3.4661$ seconds under the simulation environment described above. To identify the unknown noise parameters, we first define the search grids $\mathcal{G}_{m_1}$ and $\mathcal{G}_{m_2}$ over the intervals $[0, 1.5]$ and $[2.5, 7]$ respectively, uniformly discretizing each into $200$ evaluation points. Subsequently, we execute the SVD-based grid search (Algorithm \ref{alg:recovery}) over the parameter space. The execution time for Algorithm \ref{alg:recovery} is $0.5070$ seconds. 

We then evaluate the minimum singular value landscape generated by this constrained grid search, depicted in Fig. \ref{fig:sv_heatmap}. As shown, the global minimum of this landscape yields the estimated noise moments $(m_1^*=1.0327, m_2^*=5.0779)$, achieving a near-zero singular value of $1.54 \times 10^{-5}$. This estimate closely coincides with the moments of the noise $(m_1=1.00, m_2=5.00)$. While there are multiple feasible points where $\sigma_{\min} < \epsilon_\sigma = 0.001$ (indicated by red dots), the absolute minimum successfully isolates the optimal parameter pair $(m_1^\ast,m_2^\ast)$.

\begin{figure}[t]
    \centering
    \includegraphics[width=0.5\columnwidth]{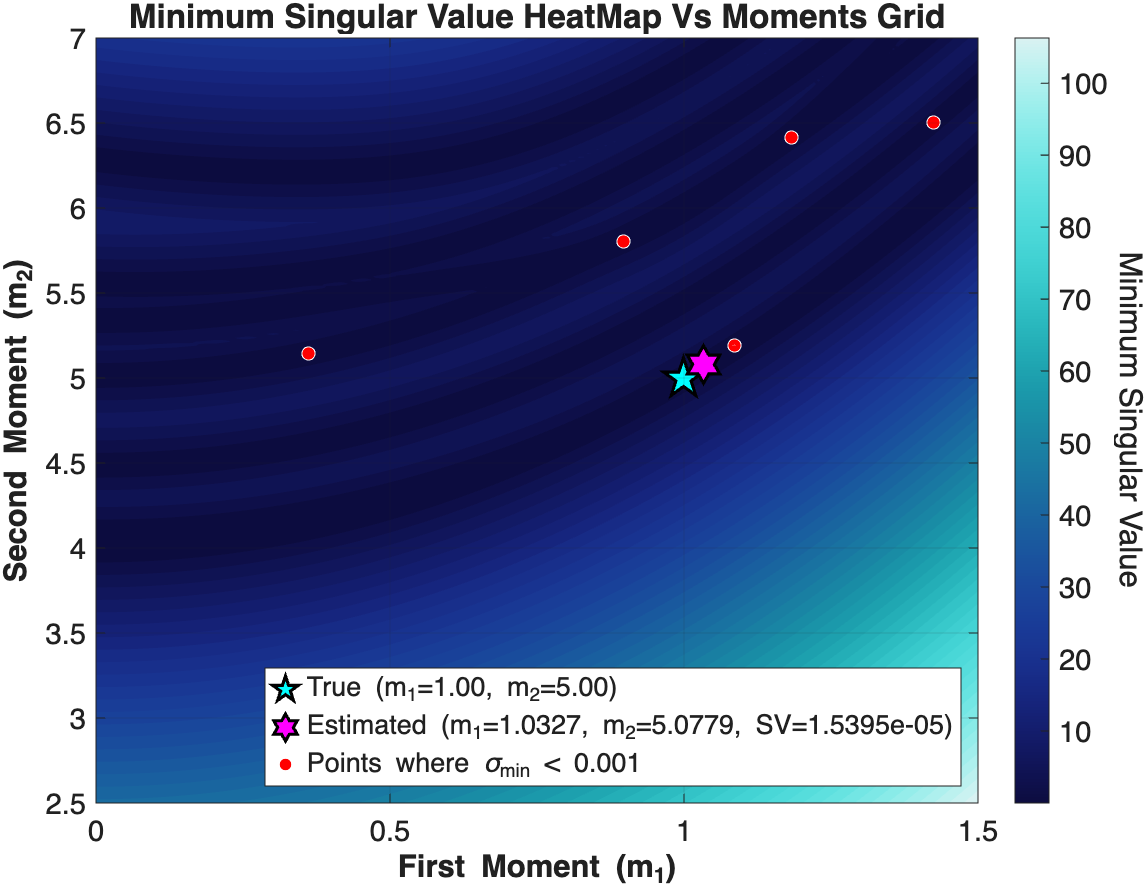}
    \caption{Heatmap of the minimum singular value ($\sigma_{\min}$). The red dots denote feasible points where $\sigma_{\min} < \epsilon_\sigma = 0.001$, with the absolute minimum (magenta star) successfully approximating the noise moments (cyan star).}
    \label{fig:sv_heatmap}
\end{figure}

To further validate this optimal estimate, we examine the numerical rank of $\mathbf{M}_\mathcal{D}(m_1, m_2)$ across the same search space, shown in Fig. \ref{fig:rank_heatmap}. For this system, the stacked Hankel matrix has $(m+p)L = 10$ block rows. Based on Willems' Fundamental Lemma, the theoretical noiseless rank dictates a complexity of $mL + n = 7$. As expected, both the noise parameters and our optimal estimated moments reside perfectly within this target rank-deficient basin (dark blue).

\begin{figure}[t]
    \centering
    \includegraphics[width=0.5\columnwidth]{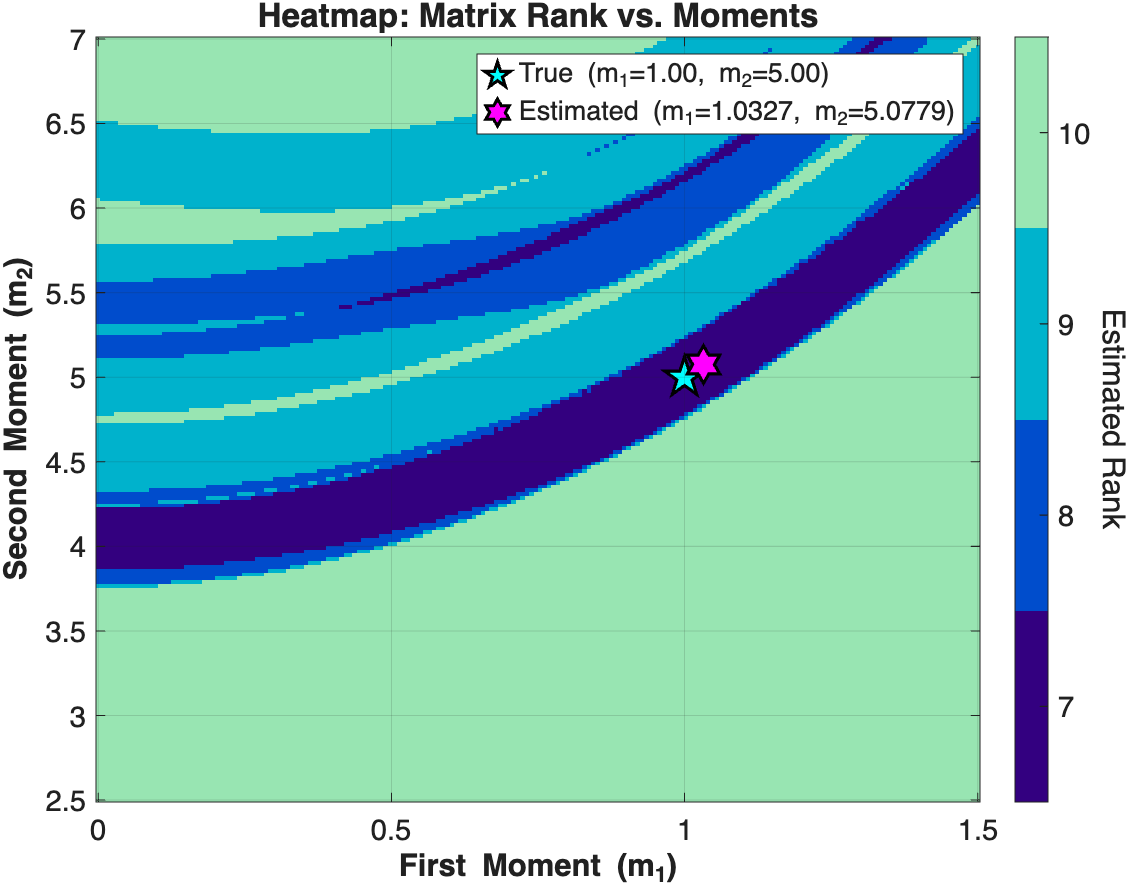}
    \caption{Heatmap of the numerical rank of $\hat{\mathbf{M}}_\mathcal{D}(m_1, m_2)$ across the search grid.}
    \label{fig:rank_heatmap}
\end{figure}

To mathematically verify that the recovered $pL-n = $3-dimensional null space represents the exact behavioral invariants of the underlying LTI system, we evaluated the geometric alignment between the estimated null space, $\hat{V}$ (using Algorithm \ref{alg:recovery}), and the underlying left null space across varying sample sizes $N_t$. 

Because the SVD yields arbitrary orthonormal bases for the identified subspaces, a direct element-wise comparison of the basis vectors is not viable. Instead, we utilize the principal angles to measure the geometric overlap between the subspaces. Let $V_{\text{true}} \in \mathbb{R}^{3 \times 10}$ denote a matrix whose rows form an orthonormal basis for the true left null space derived from the exact, noiseless system matrices. Similarly, let $\hat{V} \in \mathbb{R}^{3 \times 10}$ represent the orthonormal basis of the estimated left null space recovered by extracting singular vectors corresponding to the smallest singular values of our optimally corrected matrix $\hat{\mathbf{M}}_\mathcal{D}(m_1^*, m_2^*)$. The maximum subspace error angle is defined as: $\theta_{\max} = \arccos\left(\sigma_{\min}(V_{\text{true}} \hat{V}^\top)\right)$,
where $\sigma_{\min}(\cdot)$ denotes the minimum singular value of the resulting $3 \times 3$ inner product matrix. Since the rows of $V_{\text{true}}$ and $\hat{V}$ form orthonormal bases, the singular values of this inner product matrix strictly fall within $[0, 1]$, ensuring the domain of the $\arccos$ function is strictly valid.

\begin{figure}[t]
    \centering
    \includegraphics[width=0.5\columnwidth]{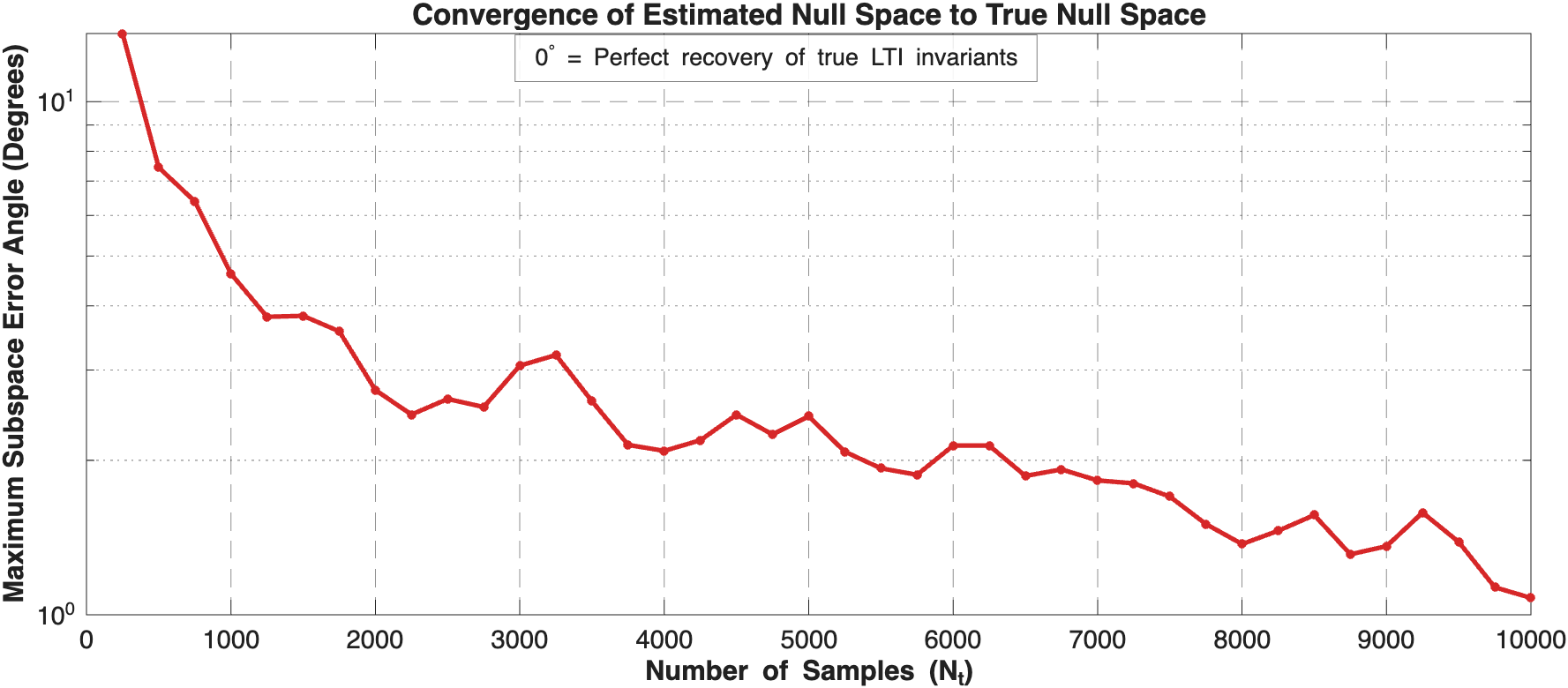}
    \caption{Logarithmic convergence of the maximum subspace error angle as the number of samples $N_t$ increases, demonstrating recovery of the underlying system invariants.}
    \label{fig:error_convergence}
\end{figure}

Fig. \ref{fig:error_convergence} demonstrates the evolution of this maximum subspace error angle as the number of experiments $N_t$ grows. Plotted on a logarithmic scale, the error trajectory exhibits a clear decay. As $N_t$ increases, the empirical noise moments rigorously converge to their expected values via the LLNs, causing the geometric gap between the true and estimated hyperplanes to strictly collapse.

\section{Conclusion}
In this paper, we enable the applicability of Willems' Fundamental Lemma to large, fragmented datasets corrupted by input-output noise, without requiring prior knowledge of the noise distributions. To achieve this, we first established a multi-experiment framework that naturally accommodates data gaps and formulated a relationship between the expected noisy measurements and noise moments. By leveraging the LLNs, we constructed an empirical estimator matrix that asymptotically converges to the uncorrupted correlation matrix. Consequently, our computationally lightweight algorithm simultaneously identifies the unknown noise moments and successfully estimates the left null space obscured by the noise. The simulation results demonstrated the effectiveness of the proposed method, confirming both its computational efficiency and its ability to accurately estimate the geometric subspace of the underlying LTI system.

%%%%%%%%%%%%%%%%%%%%%%%%%%%%%%%%%%%%%%%%%%%%%%%%%%%%%%%%%%%%%%%%%%%%%%%%%%%%%%%%
\section*{APPENDIX}\label{Appen}

\subsection{Proof of Lemma \ref{lem:noiseless_kernel}}
\begin{proof}
Substituting the definition of $\mathbf{M}_{\mathcal{D}} := \frac{1}{N_t}\sum_{i=0}^{N_t-1} \mathbb{H}^{(i)} (\mathbb{H}^{(i)})^\top$ into the quadratic form yields:
\begin{equation} \label{eq:quad_form}
\begin{aligned}
v^{\top} \mathbf{M}_{\mathcal{D}} v &= \frac{1}{N_t} \sum_{i=0}^{N_t-1} v^{\top} \mathbb{H}^{(i)} (\mathbb{H}^{(i)})^\top v \\
&= \frac{1}{N_t} \sum_{i=0}^{N_t-1} \left\| (\mathbb{H}^{(i)})^\top v \right\|_2^2.
\end{aligned}
\end{equation}

Assume $v$ lies in the intersection of all left null spaces. By definition, this means $v^{\top} \mathbb{H}^{(i)} = \mathbf{0}^{\top}$, which implies $(\mathbb{H}^{(i)})^\top v = \mathbf{0}$ for all experiments $i \in \{0, 1, \dots, N_t-1\}$. Substituting this directly into (\ref{eq:quad_form}) yields:
\begin{equation}
v^{\top} \mathbf{M}_{\mathcal{D}} v = \frac{1}{N_t} \sum_{i=0}^{N_t-1} \| \mathbf{0} \|_2^2 = 0.
\end{equation}

Conversely, assume $v^{\top} \mathbf{M}_{\mathcal{D}} v = 0$. From (\ref{eq:quad_form}), this requires:
\begin{equation}
\begin{aligned}
& \frac{1}{N_t} \sum_{i=0}^{N_t-1} \left\| (\mathbb{H}^{(i)})^\top v \right\|_2^2 = 0 \\
\implies & \left\| (\mathbb{H}^{(i)})^\top v \right\|_2^2 = 0 \implies v^{\top} \mathbb{H}^{(i)} = \mathbf{0}^\top, \quad \forall i.
\end{aligned}
\end{equation}
This establishes that $v$ belongs to the left null space of every individual experiment's Hankel matrix. Thus, $v \in \bigcap_{i=0}^{N_t-1} \mathcal{N}(\mathbb{H}^{(i)})$.
\end{proof}

\subsection{Proof of Theorem \ref{thm:lln_convergence}}
\begin{proof}
To prove the almost sure convergence of the empirical estimator $\hat{\mathbf{M}}_\mathcal{D}$, we analyze the empirical components that construct it. Specifically, for any row indices $r, s$, the estimator is built from the sample averages $\frac{1}{N_t}\sum_{i=0}^{N_t-1} \tilde{h}_r^{(i)}$ and $\frac{1}{N_t}\sum_{i=0}^{N_t-1} \tilde{h}_r^{(i)} (\tilde{h}_s^{(i)})^\top$.

Under Assumption \ref{noise_bound}, the input-output data are bounded, and measurement noise sequences possess finite moments up to the fourth order. Consequently, the random variables $\tilde{h}_r^{(i)}$ and $\tilde{h}_r^{(i)} (\tilde{h}_s^{(i)})^\top$ have variances that are uniformly bounded across all experiments. Let us define a uniform bound $\sigma_{r,s}^2$ such that for the scalar cross-correlation:
\begin{equation}
\text{Var}\left[ \tilde{h}_r^{(i)} (\tilde{h}_s^{(i)})^\top \right] \le \sigma_{r,s}^2 \quad \text{for } i = 0, 1, \dots, N_t-1.
\end{equation}

Because the dataset $\tilde{\mathcal{D}}$ consists of multiple experiments, the realizations for experiments $i$ and $j$ are strictly independent for all $i \neq j$. While they are not identically distributed (since the deterministic trajectories $h_r^{(i)}$ vary across experiments), the uniform bound on their variances strictly satisfies the condition $\sum_{i=1}^{\infty} \frac{\text{Var}(\cdot)}{i^2} < \infty$. Therefore, by Kolmogorov's Strong Law of Large Numbers (SLLN) for independent random variables \cite{grimmett2020probability}, the difference between the empirical sample average and the average of their true expected values converges almost surely to zero as $N_t \to \infty$:
\begin{equation}
\frac{1}{N_t} \sum_{i=0}^{N_t-1} \tilde{h}_r^{(i)} (\tilde{h}_s^{(i)})^\top - \frac{1}{N_t} \sum_{i=0}^{N_t-1} \mathbb{E}\left[\tilde{h}_r^{(i)} (\tilde{h}_s^{(i)})^\top\right] \to 0 \quad \text{a.s.}
\end{equation}
(An identical convergence holds for the first moment average $\frac{1}{N_t}\sum_{i=0}^{N_t-1} \tilde{h}_r^{(i)}$).

To establish the speed of convergence, following the results in the proof of Kolmogorov's SLLN \cite{grimmett2020probability}, for any $\epsilon > 0$ and integer $J$:
\begin{equation}
\begin{aligned}
&\text{Prob} \Bigg\{ \max_{N_t \ge J} \Bigg| \frac{1}{N_t} \sum_{i=0}^{N_t-1} \tilde{h}_r^{(i)} (\tilde{h}_s^{(i)})^\top \\
&\quad - \frac{1}{N_t} \sum_{i=0}^{N_t-1} \mathbb{E}\left[\tilde{h}_r^{(i)} (\tilde{h}_s^{(i)})^\top\right] \Bigg| > \epsilon \Bigg\} \\
&\le \frac{\sigma_{r,s}^2}{\epsilon^2} \left( \frac{1}{J} + \sum_{q \ge J+1} q^{-2} \right).
\end{aligned}
\end{equation}

Since the right-hand summation converges, the probability bound scales as $\mathcal{O}(1/J)$. Furthermore, because $\hat{\mathbf{M}}_\mathcal{D}$ is an affine combination of these convergent sample means, it asymptotically recovers the noiseless inner products established in Section IV, guaranteeing entry-wise almost sure convergence. Thus, we conclude:
\begin{equation}
\lim_{N_t \to \infty} \left( \hat{\mathbf{M}}_\mathcal{D} - \mathbf{M}_\mathcal{D} \right) = \mathbf{0} \quad \text{a.s.}
\end{equation}
with an asymptotic convergence rate of $\mathcal{O}(1/N_t)$.
\end{proof}
%%%%%%%%%%%%%%%%%%%%%%%%%%%%%%%%%%%%%%%%%%%%%%%%%%%%%%%%%%%%%%%%%%%%%%%%%%%%%%%%

\bibliographystyle{IEEEtran}
\bibliography{Reference}

\end{document}